\documentclass[reqno,a4paper]{amsart}
\usepackage{comment,amssymb,latexsym,upref,enumerate,fouridx}
\usepackage{mathrsfs,color}
\usepackage{centernot}

\usepackage[colorlinks,linkcolor=blue,citecolor=blue]{hyperref}
\allowdisplaybreaks
\numberwithin{equation}{section}

\theoremstyle{plain}
\newtheorem{thm}{Theorem}[section]
\newtheorem{lem}[thm]{Lemma}
\newtheorem{prop}[thm]{Proposition}

\theoremstyle{definition}

\theoremstyle{remark}
\newtheorem{rem}[thm]{Remark}

\newcommand{\ep}{\epsilon}

\newcommand{\eqn}[2]{ \begin{equation*} #2  \end{equation*} }

\input Tdef.def

\newcommand{\phib}{\bar{\phi}{}}
\newcommand{\phiv}{\vec{\phi}{}}

\begin{document}

\title[Stabilizing relativistic fluids on spacetimes with non-accelerated expansion]{Stabilizing relativistic fluids on spacetimes with non-accelerated expansion}

\author[D. Fajman]{David Fajman}
\address{Gravitational Physics\\
Faculty of Physics \\
Boltzmanngasse 5\\
University of Vienna, 1090 Vienna\\ Austria}
\email{David.Fajman@univie.ac.at}

\author[T.A. Oliynyk]{Todd A. Oliynyk}
\address{School of Mathematics\\
9 Rainforest Walk\\
Monash University, VIC 3800\\ Australia}
\email{todd.oliynyk@monash.edu}

\author[Z. Wyatt]{Zoe Wyatt}
\address{Maxwell Institute for Mathematical Sciences\\
School of Mathematics\\
University of Edinburgh, EH9 3FD, UK}
\email{zoe.wyatt@ed.ac.uk}

\begin{abstract}
\noindent
We establish global regularity and stability for the irrotational relativistic Euler equations with equation of state $\overline{p}=K\overline{\rho}$, where $0<K<1/3$, for small initial data in the expanding direction of FLRW spacetimes of the form $(\mathbb R\times\mathbb T^3,-d\tb^2+\tb^2\delta_{ij} dx^i dx^j$). This provides the first case of non-dust fluid stabilization by spacetime expansion where the expansion rate is of power law type but non-accelerated. In particular, the time integral of the inverse scale factor diverges as $t\rightarrow\infty$.
\end{abstract}

\maketitle

\section{Introduction\label{intro}}
In this article, we consider the relativistic Euler equations
\begin{align}
\vb^\mu \nablab_\mu \rhob + (\rhob+\pb) \nablab_\mu \vb^\mu &= 0, \label{eul.1} \\
(\rhob+\pb)\vb^\mu \nablab_\mu \vb^\nu + \hb^{\mu\nu}\nablab_\mu \pb &= 0, \label{eul.2}
\end{align}
with a linear equation of state
\begin{equation} \label{eos}
\pb = K \rhob, \quad 0<K<1/3,
\end{equation}
on Milne-like spacetimes of the form
\begin{equation*}
((0,\infty)\times \mathbb T^3,\gb=-d\bar t^2+\bar t^2\delta_{ij}dx^idx^j),
\end{equation*}
where $\nablab_\mu$ is the Levi-Civita connection of $\gb$,  the fluid four-velocity $\vb^\mu$ in normalized according to $\gb_{\mu\nu}\vb^\mu \vb^\nu=-1$, and 
$\hb^{\mu\nu} = \gb^{\mu\nu}+\vb^\mu \vb^\nu$ defines a positive definite inner product on the subspace of the cotangent space that is $\gb$-orthogonal to $\vb_\mu = \gb_{\mu\nu}\vb^\mu$.

We show that the canonical homogeneous solutions \eqref{homsol} to \eqref{eul.1}-\eqref{eul.2} are nonlinearly stable in the expanding direction of spacetime, in the sense that sufficiently small irrotational perturbations of these solutions exist globally towards the future and remain close to the background solutions. In particular, no shocks form in the fluids.

\subsection{Fluid regularization in expanding spacetimes}
We now discuss how the main result of this paper relates to previous work on the relativistic Euler equations. 

It is well known, due to the seminal work of Christodoulou \cite{Christodoulou:2007}, that there exist arbitrary small perturbations of the canonical constant solutions to the relativistic Euler equations with a relatively general equation of state on Minkowski spacetime that form shock singularities in finite time. In particular, these homogeneous fluid solutions are unstable.

In expanding spacetimes shock formation can be suppressed. The standard  models in cosmology representing expanding spacetimes are derived from the FLRW class, which in the case of zero spatial curvature takes the form
\begin{equation}\label{milne-like}
   ((0,\infty)\times \mathbb T^3, -d\tb^2+a(\tb)^2 \delta_{ij}dx^idx^j)\,.
\end{equation}
For such spacetimes there exists a dissipative effect on the fluid induced by the spacetime expansion that leads to the stability of homogeneous fluid solutions and thereby global regularization of the relativistic Euler equations. This effect is referred to as \emph{fluid stabilization}. Fluid stabilization was discovered by Brauer, Rendall and Reula in the Newtonian case in particular for dust $K=0$ \cite{BrauerRendallReula:1994} and rigorously established in the scenario of relativistic self-gravitating fluids in exponentially expanding spacetimes ($a(\tb)=e^{\tb}$) by Rodnianski and Speck in \cite{RodnianskiSpeck:2013}. Their result was later complemented by a series of works by Friedrich \cite{Friedrich:2017}, Oliynyk \cite{Oliynyk:CMP_2016} and Had\v{z}i\'c-Speck\cite{HadzicSpeck:2015}.

While the aforementioned results concern the case of exponentially expanding spacetimes, in \cite{Speck:2013} Speck considered the relativistic Euler equation with linear equation of state on spacetimes with a scale factor $a(\tb)$ that obeys an integrated growth condition of the form
\begin{equation*}
    \int_1^\infty a(\tb)^{-1}d\tb<\infty
\end{equation*}
or a stronger condition on the behaviour of $a(\tb)$ depending on the value of $K\in (0,1/3]$ and proved stability of homogeneous fluid solutions. In addition, he proved that for the particular case of \emph{radiation fluids} ($K=1/3$) shocks do form if the expansion rate fails to obey integrability, i.e.
if
\begin{equation*}
\int_1^\infty a(\tb)^{-1}d\tb =\infty.
\end{equation*}
In particular this implies that radiation fluid shocks would form on our Milne-like model \eqref{milne-like}.
For dust ($K=0$) only $a(\tb)^{-2}$ needs to be integrable.
In Speck's theorem an even faster expansion rate is required to stabilize \emph{massive fluids} $0< K<1/3$ compared to radiation fluids, and so it seems that stability may a priori fail to hold in general for massive non-dust fluids at the threshold where $a(\tb)^{-1}$ fails to be integrable. This threshold is interesting as it also appears independently in the context of stability of solutions to the Einstein equations as we discuss in the following.

\subsection{Localization of the Einstein equations}
In a seminal work on the Einstein-non-linear scalar field system, where the potential of the scalar field emulates a positive cosmological constant and thus generates spacetime expansion at an expontential rate, Ringström demonstrated that exponential expansion leads to a decoupling of regions of spacetime for late times \cite{Ringstrom:2008}. As a consequence, to determine the future asymptotics of solutions to the Einstein equations in a small coordinate neighborhood only the initial data in a slightly larger coordinate neighborhood is required. In this sense, the Einstein equations \textit{localize} in the presence of exponential expansion.

In a follow-up work \cite{Ringstrom:2009}, Ringstr\"om relaxed the rate of spacetime expansion to the class of power law inflation, which in our terminology corresponds to scale factors $a(\tb)=\tb^{q}$, where $q>1$ and showed that for this class the localization property still holds.  We point out that the threshold $q=1$ is precisely the one, where $\int_1^\infty a(\tb)^{-1}d\bar t$ diverges. However, the Milne model, which is a spacetime on that threshold, clearly does not possess this localization property. This is because the Milne model is a quotient of Minkowski spacetime and as a consequence no two regions can causally decouple.
We conclude that at this threshold rate of expansion the causal structure of spacetime changes drastically and this has consequences for the treatment of the Einstein equations. Nevertheless the Milne model is stable as a solution to the Einstein equations \cite{AF17,AM11}. The nature of the proof is however substantially different to the power law inflation scenario since the localization property does not hold.

\subsection{Fluid regularization in non-accelerated spacetimes}
We have now identified the threshold rate $a(\tb)=\tb$ as one between two causally different regimes from the perspective of the Einstein equations. Furthermore the result of Speck \cite{Speck:2013} shows shock formation for radiation fluids in spacetimes with this linear rate of expansion. Thus it may seem reasonable to believe that more general (i.e. non-dust and non-radiation) fluid regularization also fails at this linear rate.

However, in the context of these considerations, the result in our present paper shows that \emph{fluid regularization does occur} for zero-accelerated power law expansion ($\ddot a=0$ or $q=1$) as long as the relation $0<K<1/3$ holds. This implies that the localization property in spacetimes with accelerated expansion (or the integrability of the inverse scale factor) is not the necessary feature of the spacetime that regularizes the fluid for $0<K<1/3$, (in contrast to the case of $K=1/3$ as shown by Speck).

We conclude, that the present result establishes relativistic fluid regularization for the slowest expansion rate in comparison to previous results for the regime $0<K<1/3$. The upper bound on $K$ is sharp by Speck's result. We do not claim that the lower bound on the expansion rate is sharp. A trivial lower bound on scale factors that provide fluid regularization for $0<K<1/3$ is given by Christodoulou's result with $a(\tb)=1$. 

\subsection{Conformal rescaling of the metric and homogeneous fluid solutions}
We consider Milne-like spacetimes of the form $(M,\gb)$ where\footnote{By introducing a change of time coordinate according to the formula $\tb= 1/t$, the metric \eqref{gbdef} can
be brought into the more recognizable form
\begin{equation*}
\gb = -d\tb^2 +\tb^2 \delta_{ij}dx^idx^j,
\end{equation*}
where now $(\tb,x^i)\in [1/T_0,\infty)\times \Tbb^3$. We refer to such metrics as `Milne-like' since the scale factor is the same as in Milne, even though the spatial geometry $(\Tbb^3, \delta)$ is different from $(\mathbb{H}^3, g_{\mathbb{H}^3})$ or quotients thereof, appearing in the standard Milne spacetime. 
}
\begin{gather} 
M = (0,T_0] \times \Tbb^3, \quad T_0 >0, \label{Mdef}
\intertext{and}
\gb = \frac{1}{t^2} \biggl( -\frac{1}{t^2} dt^2  + \delta_{ij}dx^idx^j\biggr )\label{gbdef}
\end{gather}
Here, $t=x^0$ is a time coordinate on the interval $(0,T_0]$,  $(x^i)$, $i=1,2,3$, denotes standard period coordinates on the 3-torus $\Tbb^3$.
In the following,
we will use $\del{\mu}$ to denote the partial derivatives with respect to the coordinates $(x^\mu)=(t,x^i)$ and define
$\del{}^i = \delta^{ij}\del{j}$.
It is important to note that, due to our conventions, the future is located in the direction of \textit{decreasing} $t$ and future timelike infinity is located at $t=0$. 
Consequently, we require that
\begin{equation*} 
\vb^0 < 0
\end{equation*}
in order to ensure that the four-velocity is future directed.

Throughout this article, unless otherwise specified, we will assume that the constant $K$ in the linear equation of state \eqref{eos} satisfies
\begin{equation} \label{Krange}
0< K < 1/3.
\end{equation}
A straightforward calculation then shows that, for every positive constant $c_H>0$, the pair
\begin{equation} \label{homsol}
(\vb^\mu_H,\rhob_H) = \biggl(-\frac{\gb^{\mu\nu}\del{\nu}\phib_H}{\zetab_H}, \zetab_H^{\frac{1+K}{K}}\biggr),
\end{equation}
where
\begin{align}
\phib_H = c_H t^{-\lambda}, \quad \lambda =1-3K,   \label{phibHdef}
\intertext{and}
\zetab_H = \sqrt{-\gb(d\phib_H,d\phib_H)}, \label{zetabHdef}
\end{align}
defines a homogeneous, irrotational solution to the relativistic Euler equations \eqref{eul.1}-\eqref{eul.2} such that the four-velocity is future pointing. Explicitly, this solution reads

\begin{equation*}
    (\bar v^\mu_H,\bar\rho_H)=\left(-t^2\delta_0^\mu,((1-3K)c_H)^{\frac{1+K}K}t^{3(1+K)}\right).
\end{equation*}

The main aim of this article is to establish the future stability of non-linear, irrotational perturbations of these homogeneous solution
on the parameter range \eqref{Krange}. 

\subsection{Fuchsian approach}
To establish the global existence to the future of solutions to  \eqref{eul.1}-\eqref{eul.2} that represent irrotational, non-linear perturbations of the homogeneous solution \eqref{homsol}, we employ the
Fuchsian method that was introduced in \cite{Oliynyk:CMP_2016} and further developed in \cite{BOOS:2019}. This method involves transforming the relativistic Euler equations into a 
Fuchsian symmetric hyperbolic
equation of the form
\begin{equation}\label{fuchsian-class}
B^0(U)\del{t}U + \frac{1}{t}(C^i + B^i(U))\del{i} U  = \frac{1}{t}\Bc(U)\Pbb U + \frac{1}{t}F(U) \qquad \text{in $(0,T_0]\times \Tbb^{3}$,} 
\end{equation}
see Section \ref{extended} for details.
Equations of this type were studied in \cite{BOOS:2019,Oliynyk:CMP_2016} and the existence of solutions with uniform decay as $t\searrow 0$ was established 
under certain assumptions on the system's coefficients and a small initial data assumption. The Fuchsian system that we obtain for the relativistic Euler equations in this article does
not satisfy the assumptions needed to apply the existence theory from \cite{BOOS:2019,Oliynyk:CMP_2016}, as we show in Section \ref{Fuchsian}. In order to establish global existence for \eqref{fuchsian-class} we therefore generalize \cite{BOOS:2019} to this extended class in Theorem \ref{symthm}. The precise statement of the global
existence result can be found in Theorem \ref{Milnethm}.  Interestingly, while we do obtain global existence, we
do not get uniform decay as $t\searrow 0$ as was the case in \cite{BOOS:2019,Oliynyk:CMP_2016}. The obstruction to decay is the singular terms $\frac{1}{t}C^k$, where  the $C^k$, $k=1,2,3$, are constant,
symmetric matrices that are not present in the Fuchsian equations considered in \cite{BOOS:2019,Oliynyk:CMP_2016}. 

\subsection{Outlook}
Since it is known that arbitrarily small perturbations of the homogeneous solution must form shocks in finite time \cite{Speck:2013} when $K=1/3$, our results, in this sense, are sharp. It is open, whether the present result holds for spacetimes that expand slower than the Milne rate $a(\overline t)=\overline t$ for $0<K<1/3$. 
Moreover, we expect that the present result can be generalized to the rotational case. Another potential generalization concerns the inclusion of gravity by coupling the Euler equations to the Einstein equations in the framework similar to \cite{AF17,FW19}.

\subsection{Organization of the paper}
In the following section we introduce notations and setup. In Section \ref{irrot} we bring the irrotational Euler equations in the required Fuchsian form and apply the Fuchsian global existence theorem to obtain the main theorem, Theorem \ref{Milnethm}. In Section \ref{Fuchsian} the global existence theorem for Fuchsian systems of type \eqref{fuchsian-class} is proven. The appendix contains fundamental lemmas on standard functional inequalities that are used in the paper and provided here for convenience.

\section{Preliminaries\label{prelim}}
Before proceeding, we first fix our notation and introduce a few definitions that will be used throughout the article.

\subsection{Coordinates, indexing and derivatives} 
Except for Section \ref{Fuchsian}, we will use lower case Greek letters, e.g. $\alpha,\beta,\gamma$, to label spacetime coordinates indices
that run from $0$ to $3$, while we will reserve lower case Latin letters, e.g. $i,j,k$, to label spatial coordinate indices that run from $1$ to $3$.
In the appendices, we will consider general spatial dimensions, and so there, lower case Latin indices, e.g. $i,j,k$, will run from $1$ to $n$
and will be used to index spatial coordinate indices. Furthermore, $x=(x^i)$ will again denote the standard periodic coordinates, this time on
on the $n$-torus $\Tbb^n$, and $t$ will denote a time coordinate on intervals of $\Rbb$. Partial derivatives with respect to these coordinates will be denoted
by $\del{i}$ and $\del{t}$, respectively.
Additionally, in Section \ref{Fuchsian}, we use lower case Greek letters to denote multi-indices, e.g.
$\alpha = (\alpha_1,\alpha_2,\ldots,\alpha_n)\in \Zbb_{\geq 0}^n$, and we will employ the standard notation $D^\alpha = \del{1}^{\alpha_1} \del{2}^{\alpha_2}\cdots
\del{n}^{\alpha_n}$ for spatial partial derivatives and
$Du=(\del{j}u)$ for the spatial gradient. It will be clear from context whether a Greek index is meant to be interpreted as a spacetime index or a multi-index.

\subsection{Inner-products and matrix inequalities}
Throughout this article, we use
\begin{equation*}
\ipe{\xi}{\zeta} = \xi^T \zeta, \qquad \xi,\zeta \in \Rbb^N,
\end{equation*}
to denote the Euclidean inner-product and
\begin{equation*}
|\xi| = \sqrt{\ipe{\xi}{\xi}}
\end{equation*}
to denote the Euclidean norm.
Moreover, given matrices $A,B\in \Mbb{N}$, we define
\begin{equation*}
A \leq B \quad \Longleftrightarrow \quad \ipe{\xi}{A\xi} \leq \ipe{\xi}{B\xi}, \quad \forall \: \xi \in \Rbb^N,
\end{equation*}
and we use
\begin{equation*}
|A|_{\op} = \sup_{|\xi|=1}|A\xi|
\end{equation*}
to denote the operator norm of $A$.

\subsection{Sobolev spaces}
The $W^{k,p}$, $k\in \Zbb_{\geq 0}$, norm of a map $u\in C^\infty(\Tbb^N,\Rbb^N)$ is defined by
\begin{equation*}
\norm{u}_{W^{k,p}} = \begin{cases} \begin{displaystyle}\biggl( \sum_{0\leq |\alpha|\leq k} \int_{\Tbb^n} |D^{\alpha} u|^p \, d^n x\biggl)^{\frac{1}{p}}  \end{displaystyle} & \text{if $1\leq p < \infty $} \\
 \begin{displaystyle} \max_{0\leq \ell \leq k}\sup_{x\in \Tbb^n}|D^{\ell} u(x)|  \end{displaystyle} & \text{if $p=\infty$}
\end{cases}.
\end{equation*}
The Sobolev space $W^{k,p}(\Tbb^n,\Rbb^N)$ is then defined as the completion of $C^\infty(\Tbb^N,\Rbb^N)$ with respect to the norm
$\norm{\cdot}_{W^{k,p}}$. When $N=1$, we will write $W^{k,p}(\Tbb^n)$ instead, and we will employ the standard
notation $L^p(\Tbb^n)=W^{0,p}(\Tbb^n)$.
Furthermore, when $p=2$, we set $H^k(\Tbb^n,\Rbb^N)=W^{k,2}(\Tbb^n,\Rbb^N)$, and we 
use $\ip{\cdot}{\cdot}$ to denote that $L^2$ inner-product on $\Tbb^n$, that is,
\begin{equation*}
 \ip{u}{v} = \int_{\Tbb^n} \ipe{u}{v} \, d^n x.
\end{equation*}

\subsection{Constants and inequalities}
We use the standard notation
\eqn{asimb}{
a \lesssim b
}
for inequalities of the form
\eqn{aleqb}{
a \leq Cb
}
in situations where the precise value or dependence on other quantities of the constant $C$ is not required.
On the other hand, when the dependence of the constant on other inequalities needs to be specified, for
example if the constant depends on the norm $\norm{u}_{L^\infty}$, we use the notation
\eqn{cdep}{
C=C(\norm{u}_{L^\infty}).
}
Constants of this type will always be non-negative, non-decreasing, continuous functions of their arguments.

\section{Irrotational Euler equations\label{irrot}}
It is well known, see \cite[\S 3.1]{RodnianskiSpeck:2013} for example, that the irrotational relativistic Euler equations \eqref{eul.1}-\eqref{eul.2} can be formulated as a non-linear wave equation. In particular,
for the linear equation of state \eqref{eos}, the irrotational relativistic Euler equations are given by
\begin{equation} \label{awaveA}
\ab^{\alpha\beta}\nablab_\alpha \nablab_\beta \phib = 0
\end{equation}
where
\begin{equation}\label{acoustic}
\ab^{\alpha\beta} = \gb^{\alpha\beta} - \sigma \vb^\alpha \vb^\beta, \quad \sigma = \begin{textstyle}\frac{1-K}{K}\end{textstyle},
\end{equation}
is the \textit{acoustic metric}. The four-velocity is determined by
\begin{equation} \label{vbform}
\vb^\alpha = -\frac{\gb^{\alpha\beta} \del{\beta}\phib}{\zetab}
\end{equation}
where 
\begin{equation} \label{zetabdef}
\zetab = \sqrt{-\gb(d\phib,d\phib)}
\end{equation}
is the \textit{fluid enthalpy}. In this formulation, the proper energy density of the fluid can be recovered from the enthalpy via the formula
\begin{equation} \label{rhobform}
\rhob = \zetab^{\frac{1+K}{K}}.
\end{equation}

\subsection{Rescaled fluid potential}
The first step in transforming the irrotational relativistic Euler equations, given by \eqref{awaveA}, into Fuchsian form involves introducing a rescaled fluid potential via
\begin{equation} \label{phidef}
\phi = t^\lambda \phib.
\end{equation}
We then see after a straightforward calculation involving  \eqref{gbdef} and \eqref{acoustic}-\eqref{zetabdef} that the acoustic wave equation \eqref{awaveA}, when expressed in terms of $\phi$, becomes
\begin{equation} \label{awaveB}
\alpha^{00}t \del{t}(t\del{t} \phi) + 2 \alpha^{0i}t\del{t}\del{i}\phi + \alpha^{ij}\del{i}\del{j}\phi +\beta^0 t\del{t}\phi + \beta^i \del{i}\phi
+ \gamma \phi = 0
\end{equation}
where
\begin{align}
\alpha^{00} &= -1 -\frac{\sigma}{\mu}(t\del{t}\phi-\lambda \phi)^2, \label{alpha00def} \\
\alpha^{0i} & = \frac{\sigma}{\mu}(t\del{t}\phi-\lambda \phi)\del{}^i \phi, \label{alpha0jdef} \\
\alpha^{ij} & = \delta^{ij} - \frac{\sigma}{\mu}\del{}^i\phi \del{}^j \phi, \label{alphaijdef}\\
\beta^0 &= \frac{(\sigma+1)\lambda (t\del{t}\phi-\lambda \phi)^2+\lambda(\sigma-1)|D\phi|^2}{\mu}, \quad (|D\phi|^2 = \delta^{ij}\del{i}\phi\del{j}\phi), \label{beta0def} \\
\beta^i & = -\frac{2\lambda\sigma}{\mu}(t\del{t}\phi-\lambda \phi)\del{}^i \phi, \label{betaidef}\\
\gamma &= -\frac{\lambda\sigma}{\mu} |D\phi|^2 \label{gammadef}
\intertext{and}
\mu &= (t\del{t}\phi-\lambda \phi)^2-|D\phi|^2. \label{mudef}
\end{align}
For use below, we note that  $\alpha^{00}$ can be written as
\begin{equation} \label{alpha00form}
\alpha^{00} = -(1+\sigma)+ \frac{\sigma}{\mu}|D\phi|^2
\end{equation}
and that, under the rescaling \eqref{phidef}, the homogeneous solution \eqref{phibHdef} is transformed into the constant solution
\begin{equation} \label{phiHdef}
\phi_H = c_H, \quad c_H\in \Rbb_{>0},
\end{equation}
of \eqref{awaveB}.

\subsection{First order formulation}
The next step in the transformation of the irrotational relativistic Euler equations into Fuchsian form involves expressing the wave equation \eqref{awaveB} in first order form by introducing the variables
\begin{align}
\phi^{\ell}_0 &= t\del{t}\phi - \ell \delta\phi, \quad \ell \in \Rbb, \label{phi0pdef} \\
\phi_i &= \del{i}\phi \label{phiidef}
\intertext{and}
\delta\phi &= \phi-c_H. \label{deltaphidef}
\end{align}
A short calculation show that in terms of these variables, the wave equation \eqref{awaveB} is given by
\begin{align*}
-\alpha^{00} \del{t}\phi_0^\ell - \frac{2}{t} \alpha^{0i}\del{i}\phi_0^\ell -\frac{1}{t}\alpha^{ij}\del{i}\phi_j
&= \frac{1}{t}\bigl((\beta^0+\ell\alpha^{00})(\phi^\ell_0 +\ell \delta \phi)+ (\beta^i+2\ell \alpha^{0i})\phi_i +\gamma (\delta\phi+c_H) \bigr), 
\\
\alpha^{ij} \del{t}\phi_j -\frac{1}{t}\alpha^{ij}\del{j}\phi_0^\ell &= \frac{\ell}{t}\alpha^{ij}\phi_j. 
\end{align*}
We now collect together two versions of this system obtained from setting $\ell=0$ and $\ell =\lambda$. This, with the help of the identity
$\del{j}\phi^\lambda_0 = \del{j}\phi^0_0 - \lambda \phi_j$,
gives
rise to the following system
\begin{align}
-\alpha^{00} \del{t}\phi_0^0 - \frac{2}{t} \alpha^{0i}\del{i}\phi_0^\lambda -\frac{1}{t}\alpha^{ij}\del{i}\phi_j
&= \frac{1}{t}\bigl(\beta^0\phi^0_0+ (\beta^i+2\lambda\alpha^{0i})\phi_i + \gamma (\delta\phi+c_H) \bigr), \label{awaveD.1}\\
\alpha^{ij} \del{t}\phi_j -\frac{1}{t}\alpha^{ij}\del{j}\phi_0^0 &= 0, \label{awaveD.2}\\
-\alpha^{00}\del{t}\phi_0^\lambda - \frac{2}{t} \alpha^{0i}\del{i}\phi_0^0 -\frac{1}{t}\alpha^{ij}\del{i}\phi_j
&= \frac{1}{t}\bigl((\beta^0+\lambda\alpha^{00})(\phi^\lambda_0 +\lambda \delta \phi)+ \beta^i\phi_i + \gamma (\delta\phi+c_H) \bigr), \label{awaveD.3}\\
\alpha^{ij} \del{t}\phi_j -\frac{1}{t}\alpha^{ij}\del{j}\phi_0^\lambda &= \frac{\lambda}{t}\alpha^{ij}\phi_j, \label{awaveD.4}
\end{align}
where $\delta\phi$ can be recovered from $\phi^0_0$ and $\phi^\lambda_0$ via the formula
\begin{equation} \label{deltaphiform}
\delta\phi = \frac{\phi^0_0-\phi^\lambda_0}{\lambda}.
\end{equation}

Recalling that $\lambda = 1-3K$ and $\sigma = \frac{1-K}{K}$,
a short calculation using \eqref{alpha00def}, \eqref{beta0def}, \eqref{mudef} and \eqref{phiidef}  shows that
\begin{equation} \label{beta0+lambdaalpha00}
\beta^0+\lambda\alpha^{00} =-\frac{\lambda\sigma}{\mu}|\phiv|^2, \quad |\phiv|^2 = \delta^{ij}\phi_i \phi_j.
\end{equation} 
Setting
\begin{equation}\label{udef}
u = \begin{pmatrix} \phi^0_0 & \phi_j & \phi^\lambda_0 & \phi_j\end{pmatrix}^{\tr},
\end{equation}
we can then write the system \eqref{awaveD.1}-\eqref{awaveD.4} in matrix form as
\begin{equation} \label{awaveE}
A^0(u)\del{t}u + \frac{1}{t}(\Cc^k + A^k(u))\del{k}u = \frac{1}{t} \Ac(u) \Pi u+ \frac{1}{t}F(u)
\end{equation}
where
\begin{align}
A^0(u)&=\begin{pmatrix}-\alpha^{00} & 0 & 0& 0\\
0 & \alpha^{ij} & 0 & 0 \\
0 & 0 & -\alpha^{00} & 0\\
0 & 0 & 0 & \alpha^{ij} \end{pmatrix}, \label{A0def}\\
\Cc^k &= \begin{pmatrix}
0 & -\delta^{kj} & 0 & 0\\
-\delta^{ik} & 0 & 0& 0\\
0 & 0& 0& -\delta^{kj}\\
0 & 0&-\delta^{ik} & 0
\end{pmatrix}, \label{Cckdef} \\
A^k(u) &= \begin{pmatrix}
0 & \delta^{kj}-\alpha^{kj} & - 2\alpha^{0k} & 0\\
\delta^{ik}-\alpha^{ik} & 0 & 0& 0\\
- 2\alpha^{0k} & 0& 0& \delta^{kj}-\alpha^{kj}\\
0 & 0&\delta^{ik}-\alpha^{ik} & 0
\end{pmatrix}, \label{Akdef} \\
\Ac(u) &= \begin{pmatrix}\beta^0 & 0 & 0 & 0\\
0 &  \lambda \delta^{ij} & 0 &0 \\
0 & 0& (1+\sigma)\lambda & 0\\
0 & 0 & 0 & \lambda \alpha^{ij}
\end{pmatrix}, \label{Acdef} \\
\Pi &=\begin{pmatrix}1 & 0 & 0& 0\\
0 & 0 & 0 & 0 \\
0 & 0 & 0 & 0\\
0 & 0 & 0 & \delta^i_j \end{pmatrix} \label{Pidef}
\intertext{and}
F&= \begin{pmatrix}
(\beta^i+2\lambda\alpha^{0j})\phi_j +  \gamma (\delta\phi+c_H) \\
0 \\
(\beta^0+\lambda\alpha^{00})(\phi^\lambda_0 +\lambda \delta \phi)+ \beta^i\phi_i + \gamma (\delta\phi+c_H)\\
0
\end{pmatrix}. \label{Fdef}
\end{align}

\subsection{Coefficient properties}
It is not difficult to verify, with the help of the formulas \eqref{alpha00def}-\eqref{alpha00form}, \eqref{phi0pdef}-\eqref{deltaphidef}, \eqref{deltaphiform}
and \eqref{beta0+lambdaalpha00}, that the coefficients of \eqref{awaveE} satisfy the following:
\begin{enumerate}[(i)]
\item There exists a $R_0>0$ such that  the matrices  $A^0(u)$, $A^k(u)$ and $\Ac(u)$, and the source term $F(u)$ are smooth in $u$ for $u\in B_{R_0}(\Rbb^8)$.
Moreover, the matrices $A^0(u)$,  $A^k(u)$, $C^k$ and $\Pi$ are all symmetric and $\Pi$ defines a projection operator, that is,
\begin{equation} \label{Pi2=Pi}
\Pi^2 = \Pi.
\end{equation}

\begin{rem} In the following, we will always be able
to assume that $(u,Du)\in B_{R}(\Rbb^8\times \Rbb^{3\times 8})$ for any $R \in (0,R_0]$ since ultimately we will establish $L^\infty$ bounds on both $u$ and $Du$. 
Moreover, any of the implied constants in the $\lesssim$ signs will depend on $R_0$, which we take
to be fixed, but will be independent of $R \in (0,R_0]$. 
\end{rem}

\item The matrices $A^0(0)$ and $\Ac(u)$ satisfy 
\begin{equation} \label{A0Acprop1}
[\Pi,A^0(u)]= [\Pi,\Ac(u)] = 0
\end{equation}
and
\begin{equation}  \label{A0Acprop2}
\del{k}(\Pi^\perp \Ac(u)) = 0
\end{equation}
where
\begin{equation} \label{Piperpdef}
\Pi^\perp = \id - \Pi.
\end{equation}
Furthermore,
\begin{equation} \label{A0Acprop3}
|A^0(u)-A^0(0)|_{\op}+|\Ac(u)-\Ac(0)|_{\op} \lesssim |\phiv|^2
\end{equation}
where 
\begin{equation}
A^0(0)=\begin{pmatrix}1+\sigma & 0 & 0& 0\\
0 & \delta^{ij} & 0 & 0 \\
0 & 0 & 1+\sigma & 0\\
0 & 0 & 0 & \delta^{ij}\end{pmatrix} \AND \Ac(0)=\lambda A^0(0). \label{A0Acprop4}
\end{equation}
We also note by
\eqref{A0Acprop1} that 
\begin{equation*}
\Pi A^0(u)\Pi^\perp = \Pi^\perp A^0(u)\Pi =0.
\end{equation*}

\item The matrices $\Cc^k$, $A^k(u)$ and the source term $F(u)$ satisfy
\begin{align} 
\Pi^\perp \Cc^k \Pi^\perp &= 0, \label{Cckprop} \\
\Pi^\perp A^k(u)\Pi^\perp &=0, \label{Akprop1}\\
|\Pi^\perp A^k(u)\Pi|_{\op} &\lesssim |\phiv|,  \label{Akprop2}\\
|\Pi A^k(u)\Pi|_{\op} &\lesssim |\phiv|  \label{Akprop3}
\intertext{and}
|F(u)| &\lesssim |\phiv|^2. \label{Fprop}
\end{align}
\end{enumerate}

Next, using
\begin{equation*}
\del{t}\phi_j = \frac{1}{t}\del{j}\phi^0_0
\end{equation*}
and noting the time derivatives of $\phi^\lambda_0$ and $\phi^0_0$ can be computed from \eqref{awaveE},  it is not difficult to verify, with the help of
\eqref{A0Acprop1}-\eqref{A0Acprop3}
and \eqref{Cckprop}-\eqref{Fprop}, that
\begin{gather}
\biggl|\Pi^\perp (\del{t}(A^0(u))+\frac{1}{t}\del{k}(A^k(u))\Pi^\perp\biggr|_{\op} \lesssim  \frac{1}{t}\bigl(|\phi^0_0|^2+ |\phiv|^2 +|D\phi^0_0|^2+|D\phiv|^2\bigr) ,\label{divA1}\\
\biggl|\Pi^\perp (\del{t}(A^0(u))+\frac{1}{t}\del{k}(A^k(u))\Pi\biggr|_{\op} +\biggl|\Pi(\del{t}(A^0(u))+\frac{1}{t}\del{k}(A^k(u))\Pi^\perp\biggr|_{\op}\lesssim \frac{1}{t}\bigl(|\phi^0_0|+ |\phiv| +|D\phi^0_0|+|D\phiv|\bigr) \label{divA2}
\intertext{and}
\biggl|\Pi(\del{t}(A^0(u))+\frac{1}{t}\del{k}(A^k(u))\Pi\biggr|_{\op} \lesssim \frac{1}{t}\bigl(|\phi^0_0|+ |\phiv| +|D\phi^0_0|+|D\phiv|\bigr) \label{divA3}.
\end{gather}
Additionally, it is also clear from \eqref{A0Acprop3} and \eqref{Akprop1}-\eqref{Akprop3} that
\begin{equation}\label{OrdA} 
|\del{j}(A^0(u))|_{\op} +|\del{j}(\Ac(u))|_{\op}+|\del{j}(F(u))| \lesssim |\phiv|^2+|D\phiv|^2 
\end{equation}
and 
\begin{equation}\label{OrdB}
|\Pi^\perp \del{j}(A^k(u))\Pi |_{\op}+|\Pi\del{j}(A^k(u))\Pi^\perp|_{\op}+|\Pi \del{j}(A^k(u))\Pi |_{\op} \lesssim |\phiv|+|D\phiv|. 
\end{equation} 

\subsection{The extended system\label{extended}}
As it stands, the system \eqref{awaveE} is almost, but not quite, in the Fuchsian form that we require in order to apply the existence theory developed in Section \ref{Fuchsian}. To bring
it into the required form, we apply the differential operator $A^0 \del{l}(A^0)^{-1}$ to get
\begin{equation} \label{awaveF}
A^0(u) \del{t}\del{l}u +\frac{1}{t}(\Cc^k + A^k(u))\del{k}\del{l}u=\frac{1}{t}\Ac \Pi \del{l}u + \frac{1}{t}G_l(u,Du)
\end{equation} 
where
\begin{align} 
G_l(u,Du) = A^0(u)\Bigl[-\Bigl(-\del{l}&(A^0(u))^{-1}(\Cc^k+A^k(u))+(A^0(u))^{-1}\del{l}(A^k(u))\Bigr)\del{k}u
\notag \\
&+\del{l}\bigl( (A^0(u))^{-1}\Ac(u)\bigr)\Pi u + \del{l}\bigl((A^0(u))^{-1}F(u)\bigr)\Bigr]. \label{Gldef}
\end{align}
Since all the coefficients $A^0(u)$, $A^k(u)$, $\Ac(u)$ and $F(u)$ are smooth in $u$ for $u\in B_{R_0}(\Rbb^8)$, it is clear that 
$G_l(u,v)$ is smooth in $(u,v)$ for $(u,v)\in B_{R_0}(\Rbb^8)\times \Rbb^{3\times 8}$. 

Multiplying \eqref{Gldef} by $\Pi^\perp$, we obtain
\begin{align} 
\Pi^\perp G_l &= A^0\Bigl[-\Bigl(-\del{l}(A^0)^{-1}(\Pi^\perp\Cc^k+\Pi^\perp A^k)+(A^0)^{-1}\del{l}(\Pi^\perp A^k)\Bigr)\del{k}u
\notag \\
&\qquad +\del{l}\bigl( (A^0)^{-1}\Ac\bigr)\Pi^\perp \Pi u + \Pi^\perp\del{l}\bigl((A^0)^{-1}F\bigr)\Bigr] && \text{(by \eqref{A0Acprop1} \& \eqref{Piperpdef})} \notag \\
&=A^0\Bigl[-\Bigl(-\del{l}(A^0)^{-1}(\Pi^\perp\Cc^k+\Pi^\perp A^k)+(A^0)^{-1}\del{l}(\Pi^\perp A^k\Pi^\perp)
\notag \\
&\qquad  +(A^0)^{-1}\del{l}(\Pi^\perp A^k\Pi)\Bigr)\del{k}u + \Pi^\perp\del{l}\bigl((A^0)^{-1}F\bigr)\Bigr]
&& \text{(by \eqref{Piperpdef} and $\Pi^\perp \Pi=0$)}  \notag \\
& =  A^0\Bigl[-\Bigl(-\del{l}(A^0)^{-1}(\Pi^\perp\Cc^k+\Pi^\perp A^k)\del{k}u +(A^0)^{-1}\del{l}(\Pi^\perp A^k\Pi)\del{k}\Pi u\Bigr)  \notag \\
&\qquad+ \Pi^\perp\del{l}\bigl((A^0)^{-1}F\bigr)\Bigr]. && \text{(by \eqref{Pi2=Pi} and \eqref{Akprop1})} \label{PiperpGl}
\end{align}
Since we are assuming that
 $(u,Du)\in B_R(\Rbb^8\times \Rbb^{3\times 8})$, $R\in (0,R_0]$, we see from \eqref{udef}, \eqref{Pidef}, \eqref{OrdA}, \eqref{OrdB}, \eqref{Gldef} and \eqref{PiperpGl} that
\begin{align} 
|\Pi^\perp G_l(u,Du)| &\lesssim |\phiv|^2+|D\phiv|^2+|D\phi^0_0|^2 \label{OrdC}
\intertext{and}
|\Pi G_l(u,Du)| & \lesssim |Du|(|\phiv|+|D\phiv|). \label{OrdD}
\end{align}

The \textit{extended system} is then defined by combining \eqref{awaveE} and \eqref{awaveF} into the following system
\begin{equation} \label{awaveG}
B^0(U) \del{t}U +\frac{1}{t}(C^k + B^k(U))\del{k}U=\frac{1}{t}\Bc \Pbb U + \frac{1}{t}H(U)
\end{equation} 
where
\begin{align}
U &= \begin{pmatrix} u \\ Du \end{pmatrix}, \label{Udef}\\
B^0(U) &= \begin{pmatrix} A^0(u) & 0 \\ 0 & A^0(u) \end{pmatrix}, \label{B0def} \\
C^k &=  \begin{pmatrix} \Cc^k & 0 \\ 0 & \Cc^k \end{pmatrix}, \label{Ckdef}\\
B^k(U) &= \begin{pmatrix} A^k(u) & 0 \\ 0 & A^k(u) \end{pmatrix}, \label{Bkdef} \\
\Bc(U) &= \begin{pmatrix} \Ac(u) & 0 \\ 0 & \Ac(u) \end{pmatrix}, \label{Bcdef} \\
\Pbb &= \begin{pmatrix} \Pi & 0 \\ 0 & \Pi \end{pmatrix},\label{Pbbdef}\\
H(U) & = \begin{pmatrix} F(u) \\ G(u,Du) \end{pmatrix}\label{Hdef}
\end{align}
and we have set  
\begin{equation*}
G(u,Du) = \bigl(G_l(u,Du)\bigr).
\end{equation*}

The point of the extended system is that it is now in Fuchsian form and its coefficients satisfy the assumptions needed apply the existence theory developed in Section \ref{Fuchsian}. To see that the coefficients do in fact satisfy the required assumptions, we observe, with the help
of \eqref{udef}, \eqref{Pidef}, \eqref{A0Acprop4} and \eqref{B0def}-\eqref{Hdef}, that \eqref{A0Acprop3}, \eqref{Akprop1}-\eqref{Fprop}, \eqref{divA1}-\eqref{divA3}, 
and \eqref{OrdC}-\eqref{OrdD}
imply that
\begin{gather}
\biggl|\Pbb^\perp \Bigl(\del{t}(B^0(U))+\frac{1}{t}\del{k}(B^k(U))\Bigr)\Pbb^\perp\biggr|_{\op} \lesssim  \frac{1}{t}|\Pbb U|^2 ,\label{divB1}\\
\biggl|\Pbb^\perp \Bigl(\del{t}(B^0(U))+\frac{1}{t}\del{k}(B^k(U))\Bigr)\Pbb\biggr|_{\op}+\biggl|\Pbb\Bigl(\del{t}(B^0(U))+\frac{1}{t}\del{k}(B^k(U))\Bigr)\Pbb^\perp\biggr|_{\op} \lesssim \frac{1}{t}|\Pbb U|, \label{divB2}\\
\biggl|\Pbb\Bigl(\del{t}(B^0(U))+\frac{1}{t}\del{k}(B^k(U))\Bigr)\Pbb\biggr|_{\op} \lesssim \frac{1}{t}|\Pbb U|, \label{divB3}\\
\Pbb^\perp B^k(U)\Pbb^\perp =0, \label{Bkbnd1}\\
|\Pbb^\perp B^k(U)\Pbb|_{\op}+|\Pbb B^k(U)\Pbb^\perp|_{\op} \lesssim |\Pbb U|,  \label{Bkbnd2}\\
|\Pbb B^k(U)\Pbb|_{\op} \lesssim |\Pbb U|,  \label{Bkbnd3}\\
|\Pbb^\perp H(U)| \lesssim |\Pbb U|^2, \label{Hbnd1}\\
|\Pbb H(U)|  \lesssim |U||\Pbb U|, \label{Hbnd2}\\
|B^0(u)-B^0(0)|_{\op}+|\Bc(u)-\Bc(0)|_{\op} \lesssim |\Pbb U|^2 \label{Bbnds1}
\intertext{and}
\Pbb B^0(u)\Pbb^\perp=\Pbb^\perp B^0(u)\Pbb = 0,
\label{Bbnds1a}
\end{gather} 
where
\begin{equation} \label{Bbnds2}
B^0(0) = \begin{pmatrix} A^0(0) & 0 \\ 0 & A^0(0) \end{pmatrix} \AND \Bc(0) = \lambda B^0(0).
\end{equation}
and
\begin{equation} \label{Pbbperpdef}
\Pbb^\perp = \id -\Pbb.
\end{equation}
We see also from \eqref{Pi2=Pi}-\eqref{A0Acprop2} and \eqref{Bcdef}-\eqref{Pbbdef} that 
\begin{equation} \label{PbbBccom}
[\Pbb,\Bc(U)] = 0,
\end{equation}
\begin{equation} \label{PbbperpBc}
\del{k}(\Pbb^\perp \Bc(U))=0
\end{equation}
and 
\begin{equation} \label{Pbb2=Pbb}
\Pbb^2=\Pbb.
\end{equation}

\subsection{Global existence}
We are now ready to state and prove the main result of this article that guarantees the future stability of irrotational, non-linear perturbations of the homogeneous
solutions \eqref{homsol} to the relativistic Euler equations \eqref{eul.1}-\eqref{eul.2} on Milne-like spacetimes of the form \eqref{Mdef}-\eqref{gbdef}.

\begin{thm} \label{Milnethm}
Suppose $T_0>0$, $k \in \Zbb_{>3/2+2}$, $0<K<1/3$, $c_H>0$ and 
\begin{equation*}
(\phit_0,\phit_1)\in H^{k+2}(\Tbb^{3})\times H^{k+1}(\Tbb^3).
\end{equation*} 
Then there exists a $\delta>0$ such that if
\begin{equation*}
\norm{\phit_0-c_H}_{H^{k+2}} + T_0\norm{\phit_1}_{H^{k+1}} \leq \delta,
\end{equation*}
then there exists a 
\begin{equation*}
\phi \in \bigcap_{\ell=0}^{k+2} C^\ell\bigl((0,T_0],H^{k+2-\ell}(\Tbb^{3})\bigr)
\end{equation*}
such that $\phib= t^{3K-1}\phi$ defines a unique classical solution of the wave equation \eqref{awaveA} on $(0,T_0]\times \Tbb^3$ satisfying the initial conditions
\begin{equation*}
(\phib|_{t=T_0}, \del{t}\phib|_{t=T_0}) = \bigl(T_0^{3K-1}\phit_0,T_0^{3K-2}(T_0\phit_1 +(3K-1)\phit_0)\bigr),
\end{equation*}
which by \eqref{vbform} and \eqref{rhobform}, determines a (unique) irrotational solution of the relativistic Euler equations \eqref{eul.1}-\eqref{eul.2} on $(0,T_0]\times \Tbb^3$.
Moreover, $\phi$ is bounded by
\begin{equation*}
\norm{\phi(t)-c_H}_{H^{k+2}}^2 +\norm{t\del{t}\phi(t)}_{H^{k+1}}^2+ \int_{t}^{T_0} \frac{1}{\tau}\Bigl(\norm{\tau\del{\tau}\phi(\tau)}_{H^{k+1}}^2+ 
\norm{D\phi(\tau)}_{H^{k+1}}^2\Bigr)\, d\tau   \lesssim  \bigl(
\norm{\phit_0-c_H}_{H^{k+2}}^2+\norm{\phit_1}_{H^{k+1}}^2
\bigr)
\end{equation*}
for  all $t\in (0,T_0]$.
\end{thm}
\begin{proof}
First, we fix $T_0>0$, $k \in \Zbb_{>3/2+2}$,  $c_H>0$, $0<K<1/3$,  $\delta>0$ and choose $(\phit_0,\phit_1)\in H^{k+2}(\Tbb^{3})\times H^{k+1}(\Tbb^3)$ 
so that 
\begin{equation*}
\norm{\phit_0-c_H}_{H^{k+2}} + T_0\norm{\phit_1}_{H^{k+1}} \leq \delta.
\end{equation*}
We then know from standard existence and uniqueness theory for wave equations that there exists a
unique solution  
\begin{equation*}
\phi \in \bigcap_{\ell=0}^{k+2} C^\ell\bigl((T^*,T_0],H^{k+2-\ell}(\Tbb^{3})\bigr)
\end{equation*}
to the wave equation \eqref{awaveB} for some maximal time $T^* \in (0,T_0)$ that satisfies the
the initial conditions 
\begin{equation*}
(\phib|_{t=T_0}, \del{t}\phib|_{t=T_0}) = \bigl(T_0^{3K-1}\phit_0,T_0^{3K-2}(T_0\phit_1 +(3K-1)\phit_0)\bigr).
\end{equation*}
Furthermore, we know, from the calculations carried out in Section \ref{irrot}, that 
\begin{equation*}
U = \begin{pmatrix} u & Du\end{pmatrix}^{\tr},
\end{equation*}
where
\begin{equation*}
u =\begin{pmatrix} t\del{t}\phi & \del{j}\phi & t\del{t}\phi -\lambda (\phi-c_H) & \del{j}\phi \end{pmatrix}^{\tr}, \quad \lambda = 1-3 K,
\end{equation*}
defines a solution of \eqref{awaveG} on the time interval $(T^*,T_0]$. It is also clear from \eqref{divB1}-\eqref{Pbb2=Pbb}, and the symmetry of the
matrices $B^0(U)$ and $B^k(U)$ that, after the simple time transformation $t\mapsto -t$ and for $U$ satisfying $|U|\leq R$ with $R$ chosen sufficiently small, all of the assumptions from Section \ref{coeffassump}
will be satisfied for any choice of $\kappa \in (0,\lambda)$, and  $\lambda_a>0$, $a=1,2$, and $\beta_a>0$, $a=1,2,3,4$, chosen as small as we like. Since
\begin{equation*}
\norm{U(0)}_{H^k} \lesssim \norm{\phit_0-c_H}_{H^{k+2}} + T_0\norm{\phit_1}_{H^{k+1}} \leq \delta,
\end{equation*}
it then follows from Theorem \ref{symthm} that, for $\delta>0$ chosen sufficiently small, the maximal time of existence is $T^*=0$ and $U$ will satisfy an energy estimate of the form
\begin{equation} \label{MilnethmEnergy}
\norm{U(t)}_{H^k}^2+ \int_{t}^{T_0} \frac{1}{\tau} \norm{\Pbb U(\tau)}_{H^k}^2\, d\tau   \leq C\bigl(\delta,\delta^{-1}\bigr)\norm{U(T_0)}_{H^k}^2
\end{equation}
for all $t\in (0,T_0]$. In particular this shows that the solution $\phi$ exists on the time interval $(0,T_0]$. Additionally, since
\begin{equation*}
\norm{\Pbb U}_{H^k}^2 \approx \norm{t\del{t}\phi}_{H^{k+1}}^2+ 
\norm{D\phi}_{H^{k+1}}^2 \AND \norm{U}^2_{H^k} \approx \norm{\phi-c_H}_{H^{k+2}}^2 +\norm{t\del{t}\phi}_{H^{k+1}}^2,
\end{equation*}
the stated bound satisfied by $\phi$ is a direct consequence of the energy estimate \eqref{MilnethmEnergy}. This complete the proof of the theorem. 
\end{proof}

\section{Fuchsian initial value problems\label{Fuchsian}}

In this section, we develop an existence theory for the initial value problem (IVP) for Fuchsian equations of the form
\begin{align}
B^0(u)\del{t}u + \frac{1}{t}(C^i + B^i(u))\del{i} u  &= \frac{1}{t}\Bc(u)\Pbb u + \frac{1}{t}F(u) && \text{in $[T_0,T_1)\times \Tbb^{n}$,} \label{symivp.1} \\
u &=u_0 && \text{in $\{T_0\}\times \Tbb^{n}$,} \label{symivp.2}
\end{align}
where now $T_0 < T_1 \leq 0$ and the coefficients satisfy the assumptions set out in the following section. Since these assumptions imply, in particular, 
that \eqref{symivp.1} is symmetric hyperbolic, this evolution equation enjoys the Cauchy stability property. As a consequence, the existence of solutions to \eqref{symivp.1}-\eqref{symivp.2} when $T_1<0$ is guaranteed for sufficiently small initial data. Thus the main aim of this section will be to establish the existence and uniqueness of solutions for $T_1=0$ under a suitable smallness assumption on the initial data.

The study of the IVP for Fuchsian equations was initiated in \cite{Oliynyk:CMP_2016}, and, there, the existence and uniqueness of solutions
to \eqref{symivp.1}-\eqref{symivp.2} on intervals of the form $[T_0,0)$ was established under a small initial data assumption. Furthermore, decay estimates 
as $t\nearrow 0$ were also obtained. This existence theory was then generalized in \cite{BOOS:2019} to allow for certain coefficients to have a singular dependence on $t$, which significantly widened the applicability of this theory to establish global existence and decay results for systems of hyperbolic equations. For examples
of global existence results for a range of different hyperbolic systems that have been established using the Fuchsian method see 
\cite{BOOS:2019,LeFlochWei:2015,LiuOliynyk:2018b,LiuOliynyk:2018a,LiuWei:2019,Oliynyk:CMP_2016,Wei:2018}. 

The existence theory from \cite{BOOS:2019} does
not apply to the IVP \eqref{symivp.1}-\eqref{symivp.2} due to the appearance of the term $\frac{1}{t}C^k\del{k}u$, which does not satisfy the assumptions
needed for the existence theory. In the present paper we therefore establish a complementary theorem to \cite{BOOS:2019} which provides existence and uniqueness of solutions to \eqref{symivp.1}-\eqref{symivp.2} on time intervals $[T_0,0)$.  Interestingly, due
to the term $\frac{1}{t}C^k\del{k}u$, these solutions do not decay uniformly as $t\nearrow 0$, which is a key difference compared to the Fuchsian equations considered in \cite{BOOS:2019} whose solutions do decay uniformly. The precise statement of our existence result that is applicable to the IVP \eqref{symivp.1}-\eqref{symivp.2}
is given below in Theorem \ref{symthm}.

\subsection{Coefficient assumptions\label{coeffassump}}
We specify in the following the assumptions on the coefficients in \eqref{symivp.1}-\eqref{symivp.2}.
\begin{enumerate}[(i)]
\item The solution $u(t,x)$ is a $\Rbb^N$-valued map.
\item The matrix $\Pbb\in \Mbb{N}$ is a constant, symmetric projection
operator, that is,
\begin{equation} \label{Pbbprop}
\Pbb^2 = \Pbb,  \quad  \Pbb^{\tr} = \Pbb, \quad \del{t}\Pbb =0 \AND \del{j} \Pbb =0.
\end{equation}
For use below, we define the \textit{complementary projection operator} by
\begin{equation*}
\Pbb^\perp = \id -\Pbb.
\end{equation*}
\item The matrices $C^k \in \Mbb{N}$ are constant and symmetric, that is
\begin{equation} \label{Ckmat}
(C^k)^{\tr} = C^k, \quad \del{t} C^k =0 \AND \del{j} C^k =0.
\end{equation}
\item There exist constants  $\kappa, \gamma_1, \gamma_2 >0$ such that the maps $B^0,\Bc \in 
C^\infty(B_R(\Rbb^N),\Mbb{N})\bigr)$ satisfy
\begin{equation} \label{B0BCbnd}
\frac{1}{\gamma_1} \id \leq  B^0(v)\leq \frac{1}{\kappa} \Bc(v) \leq \gamma_2 \id
\end{equation}
for  all $v\in B_{R}(\Rbb^N)$, and the following additional properties:
\begin{gather} 
(B^0(v))^{\tr} = B^0(v), \label{B0sym}\\
[\Pbb,\Bc(v)] = 0, \label{BcPbbcom}\\
\del{k}\bigl(\Pbb^\perp \Bc(v)\bigr) = 0, \label{BcPbbperpdiff}\\
|\Pbb(B^0(v)-B^0(0))\Pbb|_{\op}+|\Pbb \Bc(v)-\Pbb \Bc(0)|_{\op} \lesssim |\Pbb v|, \label{BcPbbOrd}\\
\intertext{and}
\bigl|\Pbb^\perp\bigl(B^0(v)-B^0(0)\bigr)\Pbb^\perp\bigr|_{\op}+\bigl|\Pbb^\perp B^0(v)\Pbb\bigr|_{\op}+
\bigl|\Pbb B^0(v)\Pbb^\perp\bigr|_{\op} \lesssim |\Pbb v|^2 \label{B0Ord}
\end{gather}
for all $v \in B_{R}(\Rbb^N)$. 

It is then not difficult to see that \eqref{BcPbbcom}-\eqref{B0Ord} imply that 
\begin{gather}
[\Pbb^\perp, \Bc(v)] = [\Pbb,\Bc(v)^{-1}] = [\Pbb^\perp, \Bc(v)^{-1}] =0, \label{BcPbbcomB}\\
\del{k}(\Pbb^\perp \Bc(v)^{-1}) = 0, \label{BcPbbperpdiffB}\\
|\Pbb((B^0(v))^{-1}-(B^0(0))^{-1})\Pbb|_{\op}+\bigl|\Pbb\Bc(v)^{-1}-\Pbb\Bc(0)^{-1}\bigr|_{\op}\lesssim |\Pbb v|  \label{BcPbbOrdB}
\intertext{and}
\bigl|\Pbb^\perp\bigl((B^0(v))^{-1}-(B^0(0))^{-1}\bigr)\Pbb^\perp\bigr|_{\op}+\bigl|\Pbb^\perp(B^0(v))^{-1}\Pbb\bigr|_{\op} 
+\bigl|\Pbb (B^0(v))^{-1} \Pbb^\perp\bigr|_{\op} \lesssim |\Pbb v|^2 \label{B0OrdB}
\end{gather}
for all $v\in B_R(\Rbb^N)$. 

\item There exist constants $\lambda_a$, $a=1,2$, such that the map $F\in C^\infty(B_R(\Rbb^N),\Rbb^N)$ satisfies
\begin{align} 
|\Pbb F(v)| &\leq \lambda_1|\Pbb v| \label{F2vanish} 
\intertext{and}
|\Pbb^\perp F(v)| &\leq \frac{\lambda_2}{R} |\Pbb v|^2 \label{F2bnd.3}
\end{align}
for all $v\in B_R(\Rbb^N)$.

\item The  maps $B^k\in  C^\infty(B_R(\Rbb^N),\Mbb{N})$ satisfy
\begin{gather} 
(B^k(v))^{\tr}=B^k(v), \label{Bksym}\\
\bigl|\Pbb^\perp B^k(v)\Pbb \bigr|_{\op} +\bigl|\Pbb B^k(v) \Pbb^\perp \bigr|_{\op} \lesssim  |\Pbb v| , \label{BI2bnd.2}\\
\bigl|\Pbb^\perp  B^k(v) \Pbb^\perp  \bigr|_{\op}  \lesssim  |\Pbb v|^2 \label{BI2bnd.4}
\intertext{and}
\bigl|\Pbb B^k(v)\Pbb\bigr|_{\op} \lesssim |v| \label{BI2bnd.3}
\end{gather}
for all $v\in B_R(\Rbb^N)$.

\item  There exist constants $\beta_a \geq 0$, 
$a=1,2,3,4$, such that the map
\begin{equation*}
\Div\! B \: : \: B_R\bigl(\Rbb^N\times \Rbb^{n\times N}\bigr)  \longrightarrow \Mbb{N}
\end{equation*} 
defined by
\begin{align} 
\Div\!  B(t,v,w) &= D_v B^0(v)\cdot (B^0(v))^{-1}\Biggl(-\frac{1}{t}(C^k + B^k(v))w_k +
 \frac{1}{t}\Bc(v)\Pbb v + \frac{1}{t}F(v)
\biggr) + \frac{1}{t} D_v B^k(v)w_k \label{divBdef}
\end{align}
satisfies
\begin{align}
\bigl|\Pbb \Div \! B(t,v,w) \Pbb \bigr|_{\op} & \leq |t|^{-1}\beta_1, \label{divBbnd.1}\\
\bigl|\Pbb \Div\! B(t,v,w) \Pbb^\perp \bigr|_{\op} &\leq
\frac{|t|^{-1}\beta_2}{R}|\Pbb v|, \label{divBbnd.2}\\
\bigl|\Pbb^\perp \Div\! B(t,v,w) \Pbb \bigr|_{\op}& \leq  \frac{|t|^{-1}\beta_3}{R}|\Pbb v| \label{divBbnd.3}
\intertext{and}
\bigl|\Pbb^\perp \Div\! B(t,v,w) \Pbb^\perp \bigr|_{\op}& \leq\frac{|t|^{-1}\beta_4}{R^2}|\Pbb v|^2. \label{divBbnd.4}
\end{align}
\end{enumerate}

\begin{rem} \label{Divrem}
$\;$

\begin{enumerate}[(i)]
\item For a symmetric matrix $A\in \Mbb{N}$, we have the equality $|\Pbb^\perp A \Pbb|=|\Pbb A\Pbb^\perp|$. From this property, it is clear that
some of the assumptions above are redundant and we can always take $\beta_2=\beta_3$.  
\item For solutions  $u(t,x)$ of \eqref{symivp.1}, we have
\begin{align*} 
\Div\! B(t,u(t,x),D u(t,x)) &=\del{t}\bigl(B^0(u(t,x)\bigr)+\del{k}\biggl(\frac{1}{t} \bigl(C^k + B^k(u(t,x)\bigr)\biggr)\\
&=\del{t}\bigl(B^0(u(t,x)\bigr)+\frac{1}{t}\del{k}\bigl(B^k(u(t,x)\bigr).
\end{align*}
Furthermore, for the proof of Theorem \ref{symthm}, it will be clear that we only require that the estimates \eqref{divBbnd.1}-\eqref{divBbnd.4} hold for solutions where $w_i = \del{i}u$. This is important when we want to consider the spatially differentiated version of \eqref{symivp.1} together with \eqref{symivp.1}, which will yield
an equation of the same form as \eqref{symivp.1} for the variables $(u,w=Du)$. By considering this extended system, we will be able to relax the
assumptions \eqref{divBbnd.1}-\eqref{divBbnd.4} to
\begin{align*}
\bigl|\Pbb \Div \! B(t,v,w) \Pbb \bigr|_{\op} & \leq |t|^{-1}\beta_1, 
\\
\bigl|\Pbb \Div\! B(t,v,w) \Pbb^\perp \bigr|_{\op} &\leq
\frac{|t|^{-1}\beta_2}{R}(|\Pbb v|+|\Pbb w|), 
\\
\bigl|\Pbb^\perp \Div\! B(t,v,w) \Pbb \bigr|_{\op}& \leq  \frac{|t|^{-1}\beta_3}{R}(|\Pbb v|+|\Pbb w|) 
\intertext{and}
\bigl|\Pbb^\perp \Div\! B(t,v,w) \Pbb^\perp \bigr|_{\op}& \leq\frac{|t|^{-1}\beta_4}{R^2}(|\Pbb v|^2+|\Pbb w|^2). 
\end{align*}
\end{enumerate}
\end{rem}

\subsection{Preliminary estimates}
Before proceeding with the statement and proof of Theorem \ref{symthm}, we first establish some estimates that will be used in the proof. 
For a given $k\in \Zbb_{>n/2+1}$, we let $C_{\text{Sob}}>0 $ be the constant from Sobolev's inequality, that is,
\begin{equation} \label{Sobest}
\max\bigl\{\norm{D u(t)}_{L^\infty},\norm{u(t)}_{L^\infty}\bigr\} \leq C_{\text{Sob}}\norm{u(t)}_{H^k}.
\end{equation}

\begin{prop} \label{FdivBprop} 
Suppose  $k\in \Zbb_{>n/2+1}$, $u \in  B_{C_{\text{Sob}}^{-1}R}\bigl(H^{k}(\Tbb^n,\Rbb^N)\bigr)$, $v\in L^2(\Tbb^n,\Rbb^N)$, $F=F(t,u(x))$
and $\Div\! B = \Div\!B(t,u(x),D u(x))$. Then
\begin{equation*}
|\ip{u}{F}|\leq (\lambda_1+\lambda_2)\norm{\Pbb u}_{L^2}^2
\end{equation*}
and
\begin{align*}
|\ip{v}{\Div\! B v}| \leq |t|^{-1}\biggl(\beta_1 \norm{\Pbb v}_{L^2}^2+ \frac{\beta_2+\beta_3}{R} \norm{|v||\Pbb v| |\Pbb u|}_{L^1} + \frac{\beta_4}{R^2} \norm{|v|^2|\Pbb u|^2}_{L^1}\Bigr).
\end{align*}
\end{prop}
\begin{proof}
The estimates 
\begin{equation} \label{FdivBprop1}
|\ip{u}{\Pbb^\perp F}|\leq \lambda_2\norm{\Pbb u}_{L^2}^2
\end{equation}
and
\begin{align*}
|\ip{v}{\Div\! B v}| \leq |t|^{-1}\biggl(\beta_1 \norm{\Pbb v}_{L^2}^2+ \frac{\beta_2+\beta_3}{R} \norm{|v||\Pbb v| |\Pbb u|}_{L^1} + \frac{\beta_4}{R^2} \norm{|v|^2|\Pbb u|^2}_{L^1}\Bigr).
\end{align*}
are a direct consequence of Proposition 3.4 of \cite{BOOS:2019} and the assumptions \eqref{F2vanish}-\eqref{F2bnd.3} and
\eqref{divBbnd.1}-\eqref{divBbnd.4}. We see also from \eqref{Pbbprop} and \eqref{F2vanish} that
\begin{equation*}
|\ip{u}{\Pbb F}| = |\ip{\Pbb u}{\Pbb F}|\leq \lambda_1\norm{\Pbb u}_{L^2}^2
\end{equation*}
So from this and \eqref{FdivBprop1}, we find, with the help of the triangle inequality, that
\begin{equation*}
|\ip{u}{F}|=|\ip{u}{\Pbb F + \Pbb^\perp F}|=|\ip{u}{\Pbb F} + \ip{u}{\Pbb^\perp F}| \leq |\ip{u}{\Pbb F}|+ |\ip{u}{\Pbb^\perp F}| \leq (\lambda_1+\lambda_2)
\norm{\Pbb u}_{L^2}^2,
\end{equation*}
which completes the proof.
\end{proof}

\begin{prop} \label{Bcommprop}
Suppose  $k\in \Zbb_{>n/2+2}$, $1\leq |\alpha| \leq k$, $v\in L^2(\Tbb^n,\Rbb^N)$, $u \in  B_{C_{\text{Sob}}^{-1}R}\bigl(H^{k}(\Tbb^n,\Rbb^N)\bigr)$,
$\Bc=\Bc(u(x))$, $B^0=B^0(u(x))$ and $B^k=B^k(u(x))$. Then
\begin{gather*}
|\ip{v}{\Bc D^\alpha (\Bc^{-1}F)}| + |\ip{v}{\Bc[D^\alpha,\Bc^{-1}B^0](B^{0})^{-1}F}| \leq \Xi, \\
|\ip{v}{\Bc[D^\alpha,\Bc^{-1}B^0](B^{0})^{-1}\Bc \Pbb u}| \leq \Xi, \\
|\ip{v}{\Bc [D^\alpha,\Bc^{-1}B^k]\del{k}u}| + |\ip{v}{\Bc[D^\alpha,\Bc^{-1}B^0](B^{0})^{-1}B^k\del{k} u}|
\leq \Xi
\intertext{and}
|\ip{v}{\Bc [D^\alpha,\Bc^{-1}C^k]\del{k}u}|+ |\ip{v}{\Bc[D^\alpha,\Bc^{-1}B^0](B^{0})^{-1}C^k\del{k} u}| \leq \Xi
\end{gather*}
where
\begin{equation*}
\Xi =  C\bigl(\norm{u}_{H^k}\bigr) \Bigl(\norm{\Pbb v}_{L^2}\norm{\Pbb u}_{H^{k-1}}+\norm{v}_{L^2}\norm{\Pbb u}^2_{H^k} + 
\norm{\Pbb v}_{L^2}\norm{u}_{H^k}\norm{\Pbb u}_{H^k}\Bigr).
\end{equation*}
\end{prop}
\begin{proof}
Since the first three estimates follow directly from Proposition 3.6 of \cite{BOOS:2019} and the coefficient assumptions from Section \ref{coeffassump},
in particular, (ii),(iv), (v) and (vi), we only need to establish the last estimate. To do so, we assume that $1\leq |\alpha| \leq k$, and we observe from \eqref{Pbbprop}, \eqref{Ckmat}, \eqref{BcPbbcom} and \eqref{BcPbbcomB}-\eqref{BcPbbperpdiffB} that
\begin{align}
\ip{v}{\Bc [D^\alpha,\Bc^{-1}C^k]\del{k}u} &= \ip{\Pbb v}{\Bc [D^\alpha,\Bc^{-1}C^k]\del{k} u}
+\ip{\Pbb^\perp v}{\Bc [D^\alpha,\Bc^{-1}C^k]\del{k} u} \notag \\
 &= \ip{\Pbb v}{\Bc [D^\alpha,\Pbb\Bc^{-1}C^k]\del{k} u}
+\ip{\Pbb^\perp v}{\Bc [D^\alpha,\Pbb^\perp\Bc^{-1}C^k]\del{k}u} \notag \\
&= \ip{\Pbb v}{\Bc [D^\alpha,\Pbb\Bc^{-1}C^k]\del{k}u}. \notag
\end{align}
Applying the Cauchy-Schwartz inequality gives
\begin{equation*}
|\ip{v}{\Bc [D^\alpha,\Bc^{-1}C^k]\del{k}u}| \leq \norm{\Pbb v}_{L^2} \norm{\Bc [D^\alpha,\Pbb\Bc^{-1}C^k]\del{k}u}_{L^2}.
\end{equation*}
With the help of the H\"{o}lder's inequality, see Theorem \ref{Holder}, the commutator  estimates from Theorem \ref{calcpropB}, and  Sobolev's inequality, see Theorem \ref{Sobolev}, we 
then get
\begin{equation}\label{Bcommprop1}
|\ip{v}{\Bc [D^\alpha,\Bc^{-1}C^k]\del{k}u}| \leq \norm{\Pbb v}_{L^2} \norm{\Bc}_{H^k}\norm{D\Pbb\Bc^{-1}}_{H^{k-1}}\norm{u}_{H^{k}}.
\end{equation}
But from \eqref{BcPbbOrdB}, Sobolev's inequality and the Moser estimates from Theorem \ref{calcpropC}, we know that $\norm{D\Pbb\Bc^{-1}}_{H^{k-1}}\leq
C(\norm{u}_{H^k})\norm{\Pbb u}_{H^k}$,
and hence, we conclude from \eqref{Bcommprop1} that
 \begin{equation} \label{Bcommprop2a}
|\ip{v}{\Bc [D^\alpha,\Bc^{-1}C^k]\del{k}u}| \leq C(\norm{u}_{H^k})\norm{\Pbb v}_{L^2}\norm{\Pbb u}_{H^k}\norm{u}_{H^{k}}.
\end{equation}

Next, by \eqref{Pbbprop}, \eqref{BcPbbcom} and \eqref{BcPbbcomB}, we have
\begin{align*}
\ip{v}{\Bc[D^\alpha,\Bc^{-1}B^0](B^{0})^{-1}C^k\del{k}u} &= \ip{\Pbb v}{\Bc[D^\alpha,\Bc^{-1}B^0](B^{0})^{-1}C^k\del{k} u}
+\ip{\Pbb^\perp v}{\Bc[D^\alpha,\Bc^{-1}B^0](B^{0})^{-1}C^k
\del{k}u} \\
 &=\ip{\Pbb v}{\Bc[D^\alpha,\Bc^{-1}B^0](B^{0})^{-1}C^k\del{k} u}
+\ip{\Pbb^\perp v}{\Bc[D^\alpha,\Pbb^\perp\Bc^{-1}B^0](B^{0})^{-1}C^k
\del{k}u}.
\end{align*}
Estimating this expression as above yields
\begin{equation}\label{Bcommprop2}
|\ip{v}{\Bc[D^\alpha,\Bc^{-1}B^0](B^{0})^{-1}C^k\del{k}u}| \leq C(\norm{u}_{H^k})\Bigl(\norm{\Pbb v}_{L^2}\norm{D(\Bc^{-1}B^0)}_{H^{k-1}}\norm{u}_{H^{k}}
+ \norm{v}_{L^2}\norm{D(\Pbb^\perp\Bc^{-1}B^0)}_{H^{k-1}}\Bigr).
\end{equation}
But
\begin{align*}
|\Bc^{-1}B^0-(\Bc^{-1}B^0)|_{u=0}|\lesssim |\Pbb u|
\AND
|\Pbb^\perp\Bc^{-1}B^0-(\Pbb^\perp\Bc^{-1}B^0)|_{u=0}|\lesssim |\Pbb u|^2
\end{align*}
by  \eqref{BcPbbcom}-\eqref{B0Ord} and \eqref{BcPbbOrdB}, and consequently, by Sobolev's inequality and the Moser estimates from Theorem \ref{calcpropC},
we see that 
\begin{align*}
\norm{D(\Bc^{-1}B^0)}_{H^{k-1}}\leq C(\norm{u}_{H^k})\norm{\Pbb u}_{H^{k}}
\AND
\norm{D(\Pbb^\perp\Bc^{-1}B^0)}_{H^{k-1}}\leq C(\norm{u}_{H^k})\norm{\Pbb u}_{H^{k}}^2.
\end{align*}
Substituting these into \eqref{Bcommprop2} gives
\begin{equation}\label{Bcommprop3}
|\ip{v}{\Bc[D^\alpha,\Bc^{-1}B^0](B^{0})^{-1}C^k\del{k}u}| \leq C(\norm{u}_{H^k})\Bigl(\norm{\Pbb v}_{L^2}\norm{\Pbb u}_{H^k}\norm{u}_{H^{k}}
+ \norm{v}_{L^2}\norm{\Pbb u}_{H^{k}}^2\Bigr).
\end{equation}
Combining the inequalities \eqref{Bcommprop2a} and \eqref{Bcommprop3}, we see that the final estimate in the statement of the proposition holds, which
completes the proof.
\end{proof}
\begin{rem}
The structure of the term $\Xi$ in the above proposition plays an important role in the following proof of global existence. 
Schematically, the second bracketed term of $\Xi$ involves a quadratic term $\norm{\Pbb u}^2_{H^k}$ which will be bounded by the energy times a small coefficient coming from $\norm{v}_{L^2}$. A similar argument holds for the third term of $\Xi$. By contrast, the first term of $\Xi$ will require a more subtle analysis using Ehrling's lemma. This lemma will allow us to obtain a small coefficient at the expense of gaining derivatives and additional terms. Note that commutators involving the matrices $C^k$, which are the key new terms compared to \cite{BOOS:2019}, also lead to these more problematic terms in $\Xi$.
\end{rem}

\subsection{Global existence\label{global}}
The following theorem guarantees, under a suitable small initial data hypothesis, the existence of solutions to the the IVP \eqref{symivp.1}-\eqref{symivp.2} on the time interval  $[T_0,0)$.
The proof is similar to the first part of the proof of Theorem 3.8 from \cite{BOOS:2019} where existence is established. However, unlike Theorem 3.8 from \cite{BOOS:2019}, there
is no corresponding uniform decay estimate as $t\nearrow 0$ for solutions. This is due to the singular term $\frac{1}{t}C^k \del{k}u$ that prevents solutions 
to \eqref{symivp.1}-\eqref{symivp.2} from satisfying an estimate analogous to the one from Proposition 3.2 of \cite{BOOS:2019} unless additional assumptions on the coefficients are imposed.
\begin{thm} \label{symthm}
Suppose $k \in \Zbb_{>n/2+2}$, $u_0\in H^k(\Tbb^{n})$, assumptions (i)-(vii) from Section \ref{coeffassump} are fulfilled,
and the constants $\kappa$, $\gamma_1$, $\lambda_1$, $\lambda_2$, $\beta_1$, $\beta_2$, $\beta_3$, and $\beta_4$ from Section \ref{coeffassump} satisfy
\begin{equation*}
\kappa > \frac{1}{2}\gamma_1\biggl(\sum_{a=1}^4  \beta_{2a}+2(\lambda_1+\lambda_2)\biggr).
\end{equation*}
Then there exists a  $\delta > 0$ such that if  $\norm{u_0}_{H^k} < \delta$,
then there exists a unique solution 
\begin{equation*}
u \in C^0\bigl([T_0,0),H^k(\Tbb^{n},\Rbb^N)\bigr)\cap C^1\bigl([T_0,0),H^{k-1}(\Tbb^{n},\Rbb^N)\bigr)
\end{equation*}
of the IVP \eqref{symivp.1}-\eqref{symivp.2} that satisfies the
energy estimate
\begin{equation*}
\norm{u(t)}_{H^k}^2- \int_{T_0}^t \frac{1}{\tau} \norm{\Pbb u(\tau)}_{H^k}^2\, d\tau   \leq C\bigl(\delta,\delta^{-1}\bigr) \norm{u_0}_{H^k}^2.
\end{equation*}
\end{thm}
\begin{proof}
Fixing $k\in \Zbb_{>n/2+2}$, we obtain from standard local-in-time existence and uniqueness results for symmetric hyperbolic equations, e.g.
\cite[Ch.16 \S 1]{TaylorIII:1996}, the existence of a unique solution $u \in C^0([T_0,T^*),H^k)\cap C^1([T_0,T^*),H^{k-1})$ to
the IVP \eqref{symivp.1}-\eqref{symivp.2} for some maximal time $T^*\in (T_0,0]$. 
Then taking $R>0$ to be as in Section \ref{coeffassump}, we choose 
initial data such that 
\begin{equation*}
\norm{u(T_0)}_{H^k} < \delta
\end{equation*} 
where
\begin{equation*}
\delta \in (0,\Quarter\Rc) \AND \Rc = \min\bigl\{\frac{3 R}{4 C_{\text{Sob}}}, \frac{3 R}{4}\bigr\}.
\end{equation*}
Then either $\norm{u(t)}_{H^{k}} < \Rc$ for all $t\in [T_0,T^*)$ or there exists a first
time $T_*\in [T_0,T^*)$ such that 
\begin{equation*}
\norm{u(T_*)}_{H^{k}} =  \Rc \leq \begin{textstyle} \frac{3}{4} \end{textstyle} R.
\end{equation*}  
If the first case holds, we set $T_*=T^*$, and in either case, we observe by \eqref{Sobest} that  
\begin{equation} \label{Linfty}
\max\bigl\{\norm{Du(t)}_{L^\infty},\norm{u(t)}_{L^\infty}, \norm{u(t)}_{H^k} \bigr\} 
\leq \begin{textstyle} \frac{3}{4} \end{textstyle} R, \quad T_0\leq t < T_*.
\end{equation}
Next, applying
the differential operator $\Bc D^\alpha \Bc^{-1}$, where $|\alpha| \leq k$, to
\eqref{symivp.1} on the left yields
\begin{equation*}
B^0\del{t} D^\alpha u +\frac{1}{t}(C^i+ B^i)\del{i} D^\alpha u = \frac{1}{t}\Bc D^\alpha \Pbb u - \Bc [D^\alpha,\Bc^{-1} B^0]\del{t}u
-\frac{1}{t}\Bc[D^\alpha,\Bc^{-1}(C^i+ B^i)]\del{i} u + \frac{1}{t}\Bc D^\alpha(\Bc^{-1}F).
\end{equation*}
Using \eqref{symivp.1} to replace $\del{t}u$, we see that the above equation is equivalent to
\begin{align}
B^0\del{t} D^\alpha u + \frac{1}{t}(C^i+ B^i)\del{i} D^\alpha  u = & \frac{1}{t}\Bigl[\Bc\Pbb D^\alpha  u  - \Bc[D^\alpha,\Bc^{-1}B^0](B^0)^{-1}\Bc \Pbb u \Bigr] 
+ \frac{1}{t}\Bc[D^\alpha ,\Bc^{-1}B^0](B^0)^{-1}(C^i+ B^i)\del{i}u  \notag \\
&-\frac{1}{t}\Bc[D^\alpha ,\Bc^{-1}(C^i+ B^i)]\del{i} u
 - \frac{1}{t}\Bc[D^\alpha ,\Bc^{-1}B^0](B^0)^{-1}F 
 + \frac{1}{t}\Bc D^\alpha(\Bc^{-1}F). \label{symthm3.1}
\end{align}

In the following we will use \eqref{symthm3.1} to derive energy estimates that are well-behaved in the limit $t\nearrow 0$.  These energy estimates will be expressed 
in terms of the \text{energy norms} defined by
\begin{equation*}
\nnorm{u}^2_s =\sum_{\ell=0}^s \ip{D^\ell u}{B^0 D^\ell u}.
\end{equation*}
By \eqref{B0BCbnd}, we note that the energy $\nnorm{\cdot}^2_s$ and Sobolev $\norm{\cdot}_{H^k}$ norms are equivalent since they satisfy
\begin{equation*}
\frac{1}{\sqrt{\gamma_1}} \norm{\cdot}_{H^s} \leq \nnorm{\cdot}_s \leq \sqrt{\gamma_{2}} \norm{\cdot}_{H^s}.
\end{equation*}
We will employ this equivalence below without comment.

\bigskip

\noindent \underline{$L^2$-energy estimate:} To obtain a $L^2$-energy estimate for $u$, we set $\alpha=0$ in \eqref{symthm3.1} and then employ the usual energy identity that holds for symmetric hyperbolic equations
to get
\begin{equation} \label{symthm4}
\frac{1}{2} \del{t} \ip{u}{B^0 u} = \frac{1}{t}\ip{u}{\Bc \Pbb u}+\frac{1}{2} \ip{u}{\Div \! B u} + \frac{1}{t}\ip{u}{F},
\end{equation}
where 
\begin{equation*}
\Div\!B = \del{t}(B^0(u))+\frac{1}{t}\del{k}(B^k(u)).
\end{equation*}
Since  $t<0$, we have that
\begin{equation} \label{symthm4a}
\frac{2}{t}\ip{v}{\Bc\Pbb v}=\frac{2}{t}\ip{\Pbb v}{\Bc\Pbb v} \leq \frac{2\kappa}{t}\nnorm{\Pbb v}_0^2
\end{equation}
for any $v\in L^2(\Rbb^N)$ by \eqref{Pbbprop}, \eqref{B0BCbnd} and \eqref{BcPbbcom}. From this inequality and \eqref{Linfty}, Proposition \ref{FdivBprop}, and the energy identity \eqref{symthm4}, we deduce, with the help of  H\"{o}lder's inequality, 
the $L^2$-energy estimate
\begin{equation} \label{symthm5}
\del{t}\nnorm{u}^2_0\leq \frac{\rho_0}{t}\nnorm{\Pbb u}^2_0, \quad T_0\leq t < T_*,
\end{equation}
where 
\begin{equation}  \label{rho0pos}
\rho_0 = 2\kappa -\gamma_1\biggl[\sum_{a=1}^4 \beta_a+2(\lambda_1+\lambda_2)\biggr]>0.
\end{equation}

\bigskip

\noindent \underline{$H^k$-energy estimate:} Applying the $L^2$-energy identity, i.e. \eqref{symthm4}, to \eqref{symthm3.1} gives
\begin{equation} \label{symthm6}
\frac{1}{2} \del{t} \ip{D^\alpha u}{B^0 D^\alpha u} = \frac{1}{t}\ip{D^\alpha u}{\Bc \Pbb D^\alpha u}+\frac{1}{2} \ip{D^\alpha u}{\Div \! B D^\alpha u} + 
\ip{D^\alpha u}{G^\alpha}, \qquad
0\leq |\alpha| \leq k,
\end{equation}
where
\begin{align*}
G^\alpha = & \frac{1}{t}\Bigl(-\Bc[D^\alpha,\Bc^{-1}B^0](B^0)^{-1}\Bc \Pbb u  
+ \Bc[D^\alpha,\Bc^{-1}B^0](B^0)^{-1}(C^i+B^i)\del{i} u \notag \\
&-\Bc[D^\alpha,\Bc^{-1}(C^i+B^i)]\del{i} u
 - \Bc[D^\alpha ,\Bc^{-1}B^0](B^0)^{-1}F 
 + \Bc D^\alpha(\Bc^{-1}F)\Bigr). 
\end{align*}
From \eqref{symthm6}, we obtain, with the help of \eqref{Linfty}, \eqref{symthm4a}, Proposition \ref{FdivBprop} and H\"{o}lder's inequality, the estimate
\begin{align*}
\del{t}\nnorm{D^\alpha u}^2_0\leq& \frac{2\kappa-\gamma_1\beta_1}{t}\nnorm{D^\alpha \Pbb u}^2_0 
-\frac{\gamma_1(\beta_2+\beta_3+\beta_4)}{t}\nnorm{\Pbb u}_k\nnorm{\Pbb u}_{k-1} + 2\ip{D^\alpha u}{G^\alpha},  \qquad T_0\leq t < T_*.
\end{align*}
Using Proposition \ref{Bcommprop}, we bound the last term in the above inequality by
\begin{align*}
 \ip{D^\alpha u}{G^\alpha}\leq & -\frac{1}{t}C(\norm{u}_{H^k})\bigl(\norm{\Pbb u}_{H^k}\norm{\Pbb u}_{H^{k-1}}+\norm{u}_{H^k}\norm{\Pbb u}^2_{H^k}\bigr),
\end{align*}
and so, we have\footnote{The constant $C(\nnorm{u}_{k})$ implicitly depends on the various constants, e.g. $\gamma_1$, $\gamma_2$, $\beta_1$, etc., that
were introduced in the assumption in Section \ref{coeffassump}.}
\begin{align*}
\del{t}\nnorm{D^\alpha u}^2_0\leq& \frac{2\kappa-\gamma_1\beta_1}{t}\nnorm{D^\alpha \Pbb u}^2_0 
-\frac{1}{t}C(\nnorm{u}_{k})\bigl(\nnorm{\Pbb u}_k\nnorm{\Pbb u}_{k-1} + \nnorm{u}_{k}\nnorm{\Pbb u}^2_{k}\bigr),  \qquad T_0\leq t < T_*.
\end{align*}
Summing this over $\alpha$ for $0\leq |\alpha|\leq k$, we obtain, after an application of Young's inequality and Ehrling's lemma (Lemma \ref{Ehrling}),
the $H^k$ energy estimate
\begin{align}
\del{t}\nnorm{u}^2_k &\leq \frac{2\kappa-\gamma_1\beta_1-C(\nnorm{u}_{k})(\ep+\norm{u}_{k})}{t}\nnorm{\Pbb u}^2_k -\frac{1}{t}
c(\nnorm{u}_{k},\epsilon^{-1})\nnorm{\Pbb u}_{0}^2,   \label{symthm9}
\end{align}
which holds for any $\ep > 0$ and $t\in [T_0,T_*)$.

\bigskip

\noindent \underline{Global existence on $[T_0,0)\times \Tbb^n$:} Initially, we have
$\nnorm{u(T_0)}_k \leq \sqrt{\gamma_2}\norm{u(T_0)}_{H^k} < \delta \sqrt{\gamma_2}$, and so,
we can, for $\delta$ satisfying
\begin{equation} \label{deltafix}
0<\delta \leq \min\biggl\{\frac{\Rc}{2\sqrt{\gamma_1\gamma_2}},\frac{\Rc}{4}\bigr\},
\end{equation}
define $T_\delta \in (T_0,T_*)$ to be the first time such that $\nnorm{u(T_\delta)}_k = 2\delta \sqrt{\gamma_2}$, or if such a time does not exist,
set $T_\delta = T^*$. In either case, the inequality
\begin{equation*}
\nnorm{u(t)}_k \leq  2\delta \sqrt{\gamma_2}, \qquad T_0\leq t < T_\delta,
\end{equation*}
holds, which in turn, implies that
\begin{equation*}
\norm{u(t)}_{H^k} \leq \sqrt{\gamma_1} \nnorm{u(t)}_k \leq 2\delta \sqrt{\gamma_1\gamma_2} \leq \Rc, \qquad T_0 \leq t < T_\delta \leq T_*\leq T^*.
\end{equation*}

To proceed, we fix $\ep$ by setting $\ep = \delta\sqrt{\gamma_2}$. Substituting this into \eqref{symthm9} gives
\begin{align}\label{symthm10}
\del{t}\nnorm{u}^2_k \leq \frac{\rho_k}{t}\nnorm{\Pbb u}^2_k & -\frac{1}{t}
c(\delta,\delta^{-1})\nnorm{\Pbb u}_{0}^2, \qquad T_0 \leq t < T_\delta,
\end{align}
where
\begin{equation*}
\rho_k = 2\kappa-\gamma_1\beta_1-C(\delta)\delta.
\end{equation*}
But $\lim_{\delta \searrow 0}C(\delta)\delta = 0$ and $2\kappa-\gamma_1\beta_1 > 0$ by assumption, and consequently, we have
\begin{equation}\label{rhokpos}
\rho_k > 0
\end{equation}
provided $\delta>0$ is chosen small enough.
Furthermore, since $\rho_0> 0$ by \eqref{rho0pos}, we can add 
$\rho^{-1}_0c(\delta,\delta^{-1})$
times \eqref{symthm5} to \eqref{symthm10} to obtain the energy estimate
\begin{align*}
\del{t}\bigl(\nnorm{u}^2_k 
+\rho^{-1}_0c(\delta,\delta^{-1})\nnorm{u}_0^2 \bigr)\leq \frac{\rho_k}{t}\nnorm{\Pbb u}^2_k, \qquad T_0\leq t < T_\delta.
\end{align*}
Setting
\begin{equation} \label{Ekdef}
E_k(t) = \nnorm{u(t)}^2_k 
+\rho^{-1}_0c(\delta,\delta^{-1})\nnorm{u(t)}_0^2 
-\int_{T_0}^t \frac{\rho_k}{\tau}\nnorm{\Pbb u(\tau)}^2_k\, d\tau,
\end{equation}
we can write this energy estimate as
\begin{equation*}
\del{t} E_k \leq 0, \qquad T_0\leq t < T_\delta.
\end{equation*}
Integrating in time gives 
\begin{equation} \label{Eest1}
E_k(t) \leq E_k(0), \qquad T_0\leq t < T_\delta.
\end{equation}

With $\delta$
fixed so that  \eqref{deltafix}
and \eqref{rhokpos} hold, we choose $\delta_0 \in (0,\delta)$ and assume that the initial data is chosen so that
$\norm{u(T_0)}_{H^k} \leq \delta_0$. Then \eqref{Eest1} implies that
\begin{equation} \label{Eest2}
\nnorm{u(t)}_k \leq \sqrt{1+\rho^{-1}_0c(\delta,\delta^{-1})}\delta_0,  \qquad T_0\leq t < T_\delta.
\end{equation}
By further shrinking $\delta_0>0$, if necessary, so that
$0<\sqrt{1+\rho^{-1}_0c(\delta,\delta^{-1})}\delta_0 < \delta\sqrt{\gamma_2}$
also holds, we deduce from \eqref{Eest2} that
$\norm{u(t)}_k < \delta\sqrt{\gamma_2}$ for $T_0\leq t < T_\delta$.
From the definition of $T_\delta$ and the maximality of $T^*$, we conclude that $T_\delta=T_*=T^*=0$. Thus, we have established the global existence of solutions on $[T_0,0)\times \Tbb^n$.
Moreover, from \eqref{Ekdef}, \eqref{Eest1} and the equivalence of the norms $\norm{\cdot}_{H^k}$ and $\nnorm{\cdot}_k$, we see immediately that energy estimate
\begin{equation*} 
 \norm{u(t)}^2_{H^k}  -\int_{T_0}^t \frac{1}{\tau}\norm{\Pbb u(\tau)}^2_{H^k}\, d\tau \leq  C(\delta,\delta^{-1})\norm{u(T_0)}^2_{H^k}
\end{equation*}
holds for $T_0\leq t < 0$.
\end{proof}

\bigskip

\noindent \textit{Acknowledgements:}
This work was partially supported by the Australian Research Council grant DP170100630 and by the Swedish Research Council under grant no. 2016-06596 while the authors were in residence at Institut Mittag-Leffler in Djursholm, Sweden during the winter semester of 2019 as part of the program \textit{General Relativity, Geometry and Analysis: beyond the first 100 years after Einstein}.
We are grateful to the Institute for its support and hospitality during our stay. The author T.A.O would also like to thank the Albert Einstein Institute for its support during a visit in November, 2019 where work on this article was carried out. 
The author Z.W was supported by The Maxwell Institute Graduate School in Analysis and its
Applications, a Centre for Doctoral Training funded by the UK Engineering and Physical
Sciences Research Council (grant EP/L016508/01), the Scottish Funding Council, Heriot-Watt
University and the University of Edinburgh. D.F acknowledges support from the Austrian Science Fund (FWF) through the \emph{Project Geometric transport equations and the non-vacuum Einstein flow} (P 29900-N27).

\appendix

\section{\label{calc}Calculus inequalities}
In this appendix, we collect, for the convenience of the reader, a number of calculus inequalities that we employ in the next appendix. The proof of the following inequalities are well known and may be found, for example, in
the books \cite{AdamsFournier:2003}, \cite{Friedman:1976} and \cite{TaylorIII:1996}.

\begin{thm}{\emph{[H\"{o}lder's inequality]}} \label{Holder}
If $0< p,q,r \leq \infty$ satisfy $1/p+1/q = 1/r$, then
\begin{equation*}
\norm{uv}_{L^r} \leq \norm{u}_{L^p}\norm{v}_{L^q}
\end{equation*}
for all $u\in L^p(\Tbb^{n})$ and $v\in L^q(\Tbb^{n})$.
\end{thm}

\begin{thm}{\emph{[Sobolev's inequality]}} \label{Sobolev} Suppose
$1\leq p < \infty$ and $s\in \Zbb_{> n/p}$. Then
\begin{equation*}
\norm{u}_{L^\infty} \lesssim \norm{u}_{W^{s,p}}
\end{equation*}
for all $u\in W^{s,p}(\Tbb^{n})$.
\end{thm}

\begin{thm}{\emph{[Product and commutator estimates]}} \label{calcpropB} $\;$

\begin{enumerate}[(i)]
\item
Suppose $1\leq p_1,p_2,q_1,q_2\leq \infty$, $s\in \Zbb_{\geq 1}$, $|\alpha|=s$ and
\begin{equation*}
\frac{1}{p_1}+\frac{1}{p_2} = \frac{1}{q_1} + \frac{1}{q_2} = \frac{1}{r}.
\end{equation*}
Then
\begin{align*}
\norm{D^\alpha (uv)}_{L^r} \lesssim \norm{u}_{W^{s,p_1}}\norm{v}_{L^{q_1}} + \norm{u}_{L^{p_2}}\norm{v}_{W^{s,q_2}} \label{clacpropB.2.1}
\intertext{and}
\norm{[D^\alpha,u]v}_{L^r} \lesssim \norm{D u}_{L^{p_1}}\norm{v}_{W^{s-1,q_1}} + \norm{D u}_{
W^{s-1,p_2}}\norm{v}_{L^{q_2}}
\end{align*}
for all $u,v \in C^\infty(\Tbb^{n})$.
\item[(ii)]  Suppose $s_1,s_2,s_3\in \Zbb_{\geq 0}$, $\;s_1,s_2\geq s_3$,  $1\leq p \leq \infty$, and $s_1+s_2-s_3 > n/p$. Then
\begin{equation*}
\norm{uv}_{W^{s_3,p}} \lesssim \norm{u}_{W^{s_1,p}}\norm{v}_{W^{s_2,p}}
\end{equation*}
for all $u\in W^{s_1,p}(\Tbb^{n})$ and $v\in W^{s_2,p}(\Tbb^{n})$.
\end{enumerate}
\end{thm}

\begin{thm}{\emph{[Moser's estimates]}}  \label{calcpropC}
Suppose  $1\leq p \leq \infty$, $s\in \Zbb_{\geq 1}$, $0\leq k\leq s$, $|\alpha|=k$ and $f\in C^s(U)$, where
$U$ is open and bounded in $\Rbb$ and contains $0$, and $f(0)=0$. Then
\begin{equation*}
\norm{D^\alpha f(u)}_{L^{p}} \leq C\bigl(\norm{f}_{C^s(\overline{U})}\bigr)(1+\norm{u}^{s-1}_{L^\infty})\norm{u}_{W^{s,p}}
\end{equation*}
for all $u \in C^0(\Tbb^{n})\cap L^\infty(\Tbb^{n})\cap W^{s,p}(\Tbb^{n})$ with
$u(x) \in U$ for all $x\in \Tbb^{n}$.
\end{thm}

\begin{lem} {\emph{[Ehrling's lemma]}} \label{Ehrling}
Suppose $1\leq p < \infty$, $s_0,s,s_1\in \Zbb_{\geq 0}$, and $s_0 < s < s_1$. Then for any $\epsilon>0$ there exists a constant $C=C(\epsilon^{-1})$ such
that
\begin{equation*}
\norm{u}_{W^{s,p}} \leq \epsilon \norm{u}_{W^{s_1,p}} + C\norm{u}_{W^{s_0,p}}
\end{equation*}
for all $u\in W^{s_1,p}(\Tbb^{n})$.
\end{lem}

\bibliographystyle{amsplain}
\bibliography{fsMilne_v3}

\end{document}